\documentclass[12pt]{article}
\usepackage{amsmath}
\usepackage{amssymb}
\usepackage{epsfig}
\usepackage{amsthm}

\renewcommand{\arraystretch}{1.5}

\numberwithin{equation}{section}
\numberwithin{figure}{section}

\def\eq#1{(\ref{eq:#1})}
\def\lineup{\!\!\!\!\!\!\!\!&&}

\newcommand{\Tr}{\mathop{\rm Tr}\nolimits}

\def\d{\partial}
\def\eps{\epsilon}

\def\fraction#1#2{ { \textstyle \frac{#1}{#2} }}
\def\half{\fraction{1}{2}}

\newtheorem{theorem}{Theorem}[section]
\newtheorem*{``theorem''}{Theorem}

\begin{document}

\begin{titlepage}

\begin{center}

\vskip 1.0cm {\large \bf{The Identity String Field and}}\\ 
\vskip .2cm
{\large \bf{the Sliver Frame Level Expansion}}\\
\vskip 1.5cm
{\large Theodore Erler\footnote{Email: tchovi@gmail.com}}

\vskip 1.0cm

{\it {Institute of Physics of the ASCR, v.v.i.} \\
{Na Slovance 2, 182 21 Prague 8, Czech Republic}}\\
\vskip 1.5cm

{\bf Abstract}
\end{center}
We propose a modified version of the sliver-frame level expansion which 
gives a tool for analyzing singularities related to the identity 
string field. We apply this formalism to the newly discovered solutions of 
Masuda, Noumi, and Takahashi.

\noindent

\noindent
\medskip

\end{titlepage}

\tableofcontents

\section{Introduction}

In an interesting recent paper \cite{MNT}, Masuda, Noumi and Takahashi (MNT) 
discovered an analytic solution in open string field theory which 
seems to describe a D-brane with {\it negative} tension 
(what might be called a ``ghost D-brane''\footnote{The solutions proposed
in \cite{MNT} are different from the ``ghost D-brane'' 
solutions discussed in \cite{multiple,Integra}.} \cite{Okuda}). The solution, 
while interesting, is physically problematic. One possible way 
out, noted in \cite{MNT}, is that the ghost-brane solution has a term 
proportional to the identity string field, which might be considered 
singular. While there has been much discussion of singularities in analytic
solutions related to the sliver state (see, for example, 
\cite{multiple,Integra,lumps,dtakahashi,Hata,phantom}), so far there 
has been limited discussion of singularities related to 
the identity string field.\footnote{Some recent discussion of the identity
string field in connection with the tachyon vacuum solution $c(1-K)$ appears
in \cite{id1,id2,id3}.}  Indeed, many useful analytic solutions
in open string field theory (e.g. \cite{simple,bcc}) are quite 
``identity-like,'' and not obviously less singular than the ghost brane of MNT.

In this paper we hope to clarify this situation. 
Our results can be summarized as follows:
\begin{itemize}
\item For fields in the subalgebra of wedge states with insertions, 
we can define a version of the sliver frame level expansion 
called the {\it dual $\mathcal{L}^-$ level expansion}. It is similar
to the $\mathcal{L}_0$ level expansion of Schnabl \cite{Schnabl}, except 
fields are expanded around $K=\infty$ rather 
than $K=0$. This allows us to formulate the following criterion:
A solution is too identity-like if its {\it highest} level in the dual 
$\mathcal{L}^-$ level expansion is zero or positive.    
\item We show that solutions in the $KBc$ subalgebra \cite{Okawa,SSF1} 
can be partitioned into exactly six distinct gauge orbits, and all but two 
of these (the tachyon vacuum and the perturbative vacuum) are necessarily 
represented by solutions with unacceptable behavior with respect to the 
identity string field. This means, in particular, that there is no regular
solution for the MNT ghost-brane background in the $KBc$ subalgebra.
\end{itemize}

\section{$\mathcal{L}^-$ Level Expansion}

In this paper we work with the subalgebra of wedge states with insertions 
\cite{Schnabl,RZ_wedge,Schnabl_wedge}. This subalgebra is generated by 
taking sums and products of the string fields \cite{Okawa,SSF1}
\begin{equation}K,\ \ \ \ B,\ \ \ \ \phi_i\ \ (i=1,2,...).\end{equation}
We define $K$ and $B$ following the conventions of 
\cite{simple}. The string fields $\phi_i$ represent insertions of boundary
operators $\phi_i(z)$ in correlation functions on the cylinder\footnote{See 
appendix A of \cite{simple} for some background on 
the relation between string fields and operator insertions in correlation 
functions on the cylinder.}. Explicitly they can be defined
\begin{equation}\phi_i \equiv f_\mathcal{S}^{-1}\circ\phi_i(\half)|I\rangle,
\end{equation}
where $f_\mathcal{S}^{-1}(z)=\tan(\pi z/2)$ is the inverse of the sliver 
coordinate map \cite{RZO} and $|I\rangle$ is the identity string field (which
we will often write simply as ``$1$'' when no confusion can arise). 
We assume that the $\phi_i(z)$s have have definite scaling dimension 
$h_i$. Recall that the string fields $K$ and $B$
satisfy the identities \cite{Okawa}
\begin{equation} K = QB,\ \ \ \ [K,B]=0,\ \ \ \ B^2=0,\label{eq:KB_id}
\end{equation}
and $K$ generates the algebra of wedge states \cite{RZ_wedge}, 
in the sense that any star algebra power of the $SL(2,\mathbb{R})$ 
vacuum $\Omega\equiv|0\rangle$ can be written $\Omega^\alpha= e^{-\alpha K}$.

Consider the reparameterization generator $\mathcal{L}^-$, which is the 
BPZ odd component of the scaling generator in the sliver coordinate frame 
\cite{Schnabl,RZ}:
\begin{equation}\mathcal{L}_0\equiv f_\mathcal{S}^{-1}\circ L_0.\end{equation}
Specifically
\begin{equation}\mathcal{L}^- \equiv \mathcal{L}_0-\mathcal{L}_0^\star,
\end{equation}
where $^\star$ denotes BPZ conjugation. The reparameterization $\mathcal{L}^-$
generates scale transformations of the algebra of wedge states with insertions.
Acting on the elementary fields $K,B$ and $\phi_i$, it computes (twice) the
scaling dimension of the corresponding operator insertion on the cylinder:
\begin{equation}\half\mathcal{L}^- K= K,\ \ \ \ 
\ \ \ \ \half\mathcal{L}^- B = B,\ \ \ \ \ \ \ \ 
\half\mathcal{L}^-\phi_i = h_i\phi_i.\end{equation}
Then, since $\mathcal{L}^-$ is a derivation, this defines the action
of $\mathcal{L}^-$ on the whole subalgebra. 

We can decompose a string field into $\mathcal{L}^-$ eigenstates by 
expanding the field as a formal power series in $K$ around $K=0$. (To 
do this we assume that the $\phi_i$s have regular OPEs, otherwise this 
expansion can produce contact divergences.) For example, consider the zero 
momentum tachyon ground state 
\begin{equation}\sqrt{\Omega}\, c\, \sqrt{\Omega} = \frac{2}{\pi}c(0)|0\rangle,
\end{equation}
where $\sqrt{\Omega}$ is the square root of the $SL(2,\mathbb{R})$ vacuum. 
Expanding this around $K=0$ gives the expression
\begin{eqnarray}\sqrt{\Omega}\, c\, \sqrt{\Omega} \lineup
=\exp\left(-\frac{K}{2}\right)\,c\,\exp\left(-\frac{K}{2}\right)\nonumber\\
\lineup = c-\frac{1}{2}(cK+Kc)+\frac{1}{8}(K^2 c +2KcK+cK^2)-...\ \ .
\end{eqnarray}
Each term in the series is an $\mathcal{L}^-$ eigenstate. For a general 
state in the subalgebra of wedge states with insertions, expanding around $K=0$
produces an expression of the form
\begin{equation}\Phi = \Phi_{h_1}+\Phi_{h_2}+\Phi_{h_3}+...\ \ \ \ \ \ \ 
(h_1<h_2<h_3<...),
\end{equation} 
where $\Phi_{h_n}$ are eigenstates of $\mathcal{L}^-$:
\begin{equation}
\half \mathcal{L}^-\Phi_{h_n} = h_n\Phi_{h_n}.\end{equation}
This defines what we call the $\mathcal{L}^-$ {\it level expansion}. We use 
``level'' to refer to the $\half\mathcal{L}^-$ eigenvalue in this expansion.
 Assuming that products of $\phi_i$s do not produce operators of 
arbitrarily negative conformal dimension, the level is bounded from below,
and higher level states are considered ``subleading.'' We say ``subleading''
in quotes since the $\mathcal{L}^-$ level expansion is very formal: Each term
in the expansion is proportional to the identity string field, and there is 
at least one sense (explained in the next section) that terms in the 
expansion become increasingly singular as the level is increased. 

The $\mathcal{L}^-$ level expansion can be understood as a variant of the 
$\mathcal{L}_0$ level expansion of Schnabl \cite{Schnabl}. They are 
related through the formula
\cite{simple}
\begin{equation}\mathcal{L}_0\left(\sqrt{\Omega}\,\Phi\sqrt{\Omega}\right)
=\sqrt{\Omega}\left(
\half\mathcal{L}^-\Phi \right)\sqrt{\Omega}.\end{equation}
This means that the $\mathcal{L}_0$ level expansion of the field 
$\sqrt{\Omega}\,\Phi\sqrt{\Omega}$ is equivalent to the 
$\mathcal{L}^-$ level expansion of the field $\Phi$; The coefficients of 
the eigenstates are the same in either expansion. 
One important difference, however, is that the eigenstates of the 
$\mathcal{L}_0$ level expansion are proportional to the $SL(2,\mathbb{R})$ 
vacuum rather than the identity string field. This makes it possible to 
calculate the energy of solutions in the $\mathcal{L}_0$ level expansion 
\cite{simple,Schnabl,Aldo1,Aldo2}, whereas this is impossible in the 
$\mathcal{L}^-$ level expansion.

The $\mathcal{L}^-$ level expansion is closely related to the 
{\it phantom term} in string field theory solutions. The phantom term 
between solutions $\Phi_1$ and $\Phi_2$ can be written schematically
\begin{equation}X^\infty (\Phi_2-\Phi_1),\end{equation}
where $X^\infty$ is the boundary condition changing projector for a singular 
gauge transformation connecting $\Phi_1$ and $\Phi_2$ \cite{Integra,phantom}.
Typically, $X^\infty$ is proportional to the sliver state, and when calculating
contractions involving the phantom term the solutions $\Phi_1$ and $\Phi_2$ 
look, by comparison to $X^\infty$, like operator insertions on the 
identity string field. Moreover, in correlation functions on a very large 
cylinder, operator insertions with the lowest scaling dimension make the 
leading contribution. Therefore, inside the phantom term the solutions 
$\Phi_1$ and $\Phi_2$ are naturally described by the $\mathcal{L}^-$ 
level expansion.

Let us make a technical comment. In order to implement the $\mathcal{L}^-$ 
level expansion in the subalgebra of wedge states with insertions, we need to 
make some assumptions about how the string field depends on $K$. 
First, it must be formally an {\it analytic} function of $K$ at $K=0$, 
and second, the power series expansion in $K$ at $K=0$ must 
uniquely characterize the string field, if it exists. The 
first assumption is needed otherwise the expansion at $K=0$ does not 
necessarily produce eigenstates of $\mathcal{L}^-$. For example, the field 
$K\ln(K)$ does not have a well-defined $\mathcal{L}^-$ level expansion. 
Moreover, it is known that certain singularities at $K=0$---particularly 
poles \cite{lumps,multiple,multi_bdry}---should not be allowed for regular
string fields, though it is not clear whether complete analyticity at $K=0$ is 
required. As for the second assumption, the $\mathcal{L}^-$ level expansion
uniquely describes the string field as long as elements of the wedge algebra 
are described as Laplace transforms over wedge states 
\cite{SSF2}. More exotic functions of $K$ have been constructed in 
\cite{exotic}, but it is not clear whether such states have well-defined 
star products. At any rate, these assumptions hold for nearly all states which 
are normally considered in applications, and we will assume that they hold 
for the remainder of the paper.

\section{Dual $\mathcal{L}^-$ Level Expansion}

With this background, we are ready to discuss our main interest: understanding 
singularities in solutions related to the identity string 
field. Specifically, we want to answer the following question: 
If $\Phi$ is a string field in the subalgebra of wedge states with 
insertions, under what circumstances is the quantity\footnote{We denote 
the 1-string vertex by the trace: $\Tr[\Phi]\equiv \langle I|\Phi\rangle$} 
\begin{equation}\Tr[\Phi]\end{equation}
well-defined from the perspective of the identity string field?

To start, consider the trace of a product of local insertions
\begin{equation}\Tr[\phi_1\phi_2...\phi_n].\label{eq:unreg_phi}\end{equation}
Since the $\phi_i$s are proportional to the identity string field, the
trace represents a correlation 
function on a cylinder with vanishing area (see figure \ref{fig:IdSing1}). 
To make sense of this, we can introduce a 
point-splitting regulator, separating the $\phi_is$ with wedge states of 
non-zero width: 
\begin{equation}\Tr[\phi_1\, \Omega^{\eps t_1}\, \phi_2\, 
\Omega^{\eps t_2}\, ...
\,\phi_n\,\Omega^{\eps t_n}].\label{eq:reg_phi}\end{equation}
Formally, this reproduces \eq{unreg_phi} in the $\eps\to 0$ limit. Now note 
that \eq{reg_phi} only depends on $\eps$ through the total scaling dimension
of the insertions:
\begin{eqnarray}\Tr[\phi_1\, \Omega^{\eps t_1}\, \phi_2\, 
\Omega^{\eps t_2}\, ...
\,\phi_n\,\Omega^{\eps t_n}]\lineup 
= \left(\frac{1}{\eps}\right)^{h_1+...+h_n}
\Tr\left[\eps^{\frac{1}{2}\mathcal{L}^-}\left(\phi_1\, \Omega^{t_1}\, 
...
\,\phi_n\,\Omega^{t_3}\right)\right]\nonumber\\
\lineup = \left(\frac{1}{\eps}\right)^{h_1+h_2+...+h_n}
\Tr\left[\phi_1\, \Omega^{t_1}\, \phi_2\, \Omega^{t_2}
...
\,\phi_n\,\Omega^{t_n}\right],
\label{eq:reg_phi2}\end{eqnarray}
where we used reparameterization invariance of the vertex. 
Now let's take $\eps$ to zero. The limit is clearly divergent if 
$h_1+...+h_n$ is positive. In fact, even if $h_1+...+h_n=0$ the limit 
is singular because the answer depends on the choice of parameters 
$t_1,...,t_n$ used to split the operators. Only if $h_1+...+h_n$ is 
negative do we find a well defined result:
$\Tr[\phi_1...\phi_n]=0$. With an eye towards 
generalization, we can summarize these 
observations (in a somewhat technical fashion) as follows: 
\begin{theorem} \label{thm:pre_Lm} 
If $\Phi$ is a state created by taking sums and products 
of local insertions $\phi_i$, then $\Tr[\Phi]$ is well defined only if the
highest level in the $\mathcal{L}^-$ level expansion of $\Phi$ is strictly 
negative.\end{theorem}
\noindent At the moment this statement is somewhat uninteresting, since 
for a fixed set of $\phi_i$s there are generally few, if any, ways to create 
negative level states.

\begin{figure}
\begin{center}
\resizebox{2.5in}{2.5in}{\includegraphics{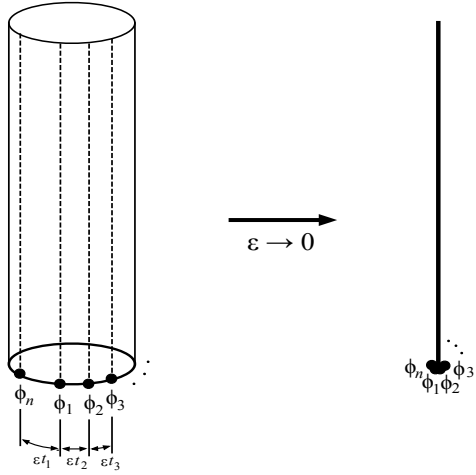}}
\end{center}
\caption{\label{fig:IdSing1} Regulating the trace of a product of local 
insertions.}\end{figure}

We would like to find an analogue of this theorem for arbitrary 
fields in the subalgebra of wedge states with insertions. As a first step, 
consider the trace of a state with a single insertion:
\begin{equation}\Tr[F(K)\phi].\label{eq:ex2}\end{equation}
Assuming that $F(K)$ can be written as a Laplace transform of some 
$f(t)$, we can write
\begin{equation}\Tr[F(K)\phi]=\int_0^\infty dt\, f(t)
\Tr[\Omega^t\phi].\end{equation}
Note that the identity string field corresponds to the $t\to 0$ limit in the 
integrand. Using a reparameterization we can factor the trace out of the 
integration:
\begin{equation}\Tr[F(K)\phi]=\Tr[\Omega\phi]
\int_0^\infty dt\, \left(\frac{1}{t}\right)^{h}f(t).\label{eq:ex22}
\end{equation}
The integration is finite near $t=0$ if
\begin{equation}\lim_{t\to 0}\left(\frac{1}{t}\right)^{h-1}f(t)=0.
\label{eq:cond1}\end{equation}
Under this condition $\Tr[F(K)\phi]$ is well-defined from the perspective
of the identity string field. This does not mean that $\Tr[F(K)\phi]$ is 
necessarily finite. For this to be true, the full integration from 
$t=0$ to $\infty$ must be finite in \eq{ex22}, not just the integration in the
neighborhood of $t=0$.

Let's restate condition \eq{cond1} in a form analogous
to theorem \ref{thm:pre_Lm}.  
To start, let's suppose that $f(t)$ can be written as a power series 
expansion around $t=0$:
\begin{eqnarray}f(t)\lineup 
= \sum_{n=1}^\infty a_n t^{\mu_n} + r(t)\nonumber\\
\lineup = a_1 t^{\mu_1} + a_2 t^{\mu_2}+...+r(t)\ \ \ \ \ \ 
(-1<\mu_1<\mu_2<....),
\label{eq:exp_t}\end{eqnarray}
where $a_n$ are coefficients, $\mu_n$ are an increasing sequence of 
powers (not necessarily integers), and $r(t)$ is a remainder which 
vanishes faster than any power at $t=0$. Condition \eq{cond1} implies that
$\mu_1$, the leading power in this expansion, satisfies the inequality
\begin{equation}\mu_1+1>h.\label{eq:cond2}\end{equation}
Now suppose that we take the expansion \eq{exp_t} to define an expansion of 
the string field $F(K)$:
\begin{equation}F(K) 
= \sum_{n=1}^\infty a_n 
\left(\int_0^\infty dt\, t^{\mu_n}\Omega^t\right) 
+R(K),\label{eq:exp_K}\end{equation}
where $R(K)$ is the Laplace transform of $r(t)$. There are two immediate 
problems with this expression. First, the series for $f(t)$ 
might only have finite radius of convergence, so \eq{exp_K} might only be 
an asymptotic expansion. Second, and perhaps more disturbingly, each term 
in \eq{exp_K} is actually a divergent string field. The integrand is 
completely unsuppressed for large $t$, so the formal expansion 
\eq{exp_K} produces a sequence of increasingly severe divergences 
proportional to the sliver state. However, for our purposes this is not a
problem. Our primary interest is the identity string 
field, not the sliver state. If the string field is well-defined, the 
sliver divergences will cancel upon formal resummation of \eq{exp_K}. 
Note that the divergent integrals in \eq{exp_K} actually define inverse 
powers of $K$ in the Schwinger parameterization:
\begin{equation}\int_0^\infty\frac{t^{\mu}}{\Gamma(\mu+1)}\Omega^t
\equiv \frac{1}{K^{\mu+1}}.\label{eq:Ltrans}\end{equation}
Therefore \eq{exp_K} can be reexpressed in the simple form
\begin{equation}F(K) = \sum_{n=1}^\infty\frac{b_n}{K^{\nu_n}}+R(K),
\label{eq:FKexp2}\end{equation}
where 
\begin{equation}\nu_n \equiv \mu_n+1,\ \ \ \ \ \ b_n \equiv 
\Gamma(\mu_n+1)a_n.\end{equation} 
It is natural to interpret \eq{FKexp2} as an expansion of $F(K)$ around 
$K=\infty$. In fact, \eq{FKexp2} looks like an expansion in $\mathcal{L}^-$ 
eigenstates, in the sense that we should define
\begin{equation}\half \mathcal{L}^- \frac{1}{K^\nu} = -\nu \frac{1}{K^\nu},
\label{eq:geneig}\end{equation}
by analogy to the formula $\half\mathcal{L}^-K^n = nK^n$ for positive integer
$n$.  By convention, we will say that $R(K)$ has $\mathcal{L}^-$ 
eigenvalue $-\infty$, since it vanishes faster than any inverse power of $K$
towards infinity. Therefore, we can interpret \eq{FKexp2} as an unusual 
form of the $\mathcal{L}^-$ level expansion.

For a general field $\Phi$ in the algebra of wedge states with insertions, a 
formal expansion around $K=\infty$ may produce an expression of the form
\begin{equation}\Phi =\Phi_{h_1}+\Phi_{h_2}+\Phi_{h_3}+...\ \ \ \ \ \ 
(h_1>h_2>h_3> ...),\end{equation}
where $\Phi_{h_n}$ are eigenstates of $\mathcal{L}^-$:
\begin{equation}\half\mathcal{L}^- \Phi_{h_n} = h_n\Phi_{h_n}.\end{equation}
This is what we call the {\it dual $\mathcal{L}^-$ level expansion}.
We use ``level'' to refer to the $\half\mathcal{L}^-$ eigenvalue in this 
expansion. The level can be taken as a precise 
measure of how identity-like a string field is. 
The identity string field itself is level $0$; progressively negative 
levels become less singular from the perspective of the identity string 
field, and progressively positive levels are increasingly more singular 
than the identity string field. Though the dual $\mathcal{L}^-$ level 
expansion looks similar to the $\mathcal{L}^-$ level expansion defined 
earlier, there are a number of important differences:
\begin{enumerate}
\item[1)] The leading level in the dual $\mathcal{L}^-$ level expansion is 
the highest level. Subleading levels are increasingly negative. In the 
$\mathcal{L}^-$ level expansion, the situation is opposite: The 
leading level is the lowest level, and subleading levels are increasingly 
positive.
\item[2)] Eigenstates in the dual $\mathcal{L}^-$ level expansion become 
increasingly singular from the perspective of the sliver state as the level 
becomes progressively negative, whereas the eigenstates in the
$\mathcal{L}^-$ level expansion become increasingly singular from the
perspective of the identity string field as the level becomes progressively 
positive.
\item[3)] Since the dual $\mathcal{L}^-$ level expansion respects the short
distance structure of the string field, it does not produce collisions of 
$\phi_i$s if they are not present in the state to begin with. 
Therefore it is not necessary to assume that the $\phi_i$s have regular OPEs 
when implementing the dual $\mathcal{L}^-$ level expansion, whereas this 
assumption appears necessary for the $\mathcal{L}^-$ (or $\mathcal{L}_0$) 
level expansion.
\end{enumerate}
\noindent 
In a sense, the two expansions give complementary 
information about the string field. For example, the states
\begin{equation}\Omega\ \ \ \ \ \ \ \Omega^{-1}\end{equation}
look very similar in the $\mathcal{L}^-$ level expansion as a power series 
around $K=0$. But in the dual $\mathcal{L}^-$ level expansion the inverse 
wedge state is clearly singular since it diverges at $K=\infty$ faster 
than any power of $K$. On the other hand, consider the states
\begin{equation}\frac{1-\Omega}{K}\ \ \ \ \ \ \ 
\frac{2-\Omega}{K}.\end{equation}
These states look similar in the dual $\mathcal{L}^-$ expansion, 
but in the $\mathcal{L}^-$ level expansion the first state is regular, 
while the second has a pole at $K=0$ which produces a sliver divergence. 
In summary, the dual $\mathcal{L}^-$ level expansion is sensitive to 
singularities related to the identity string field, but not the sliver state,
while precisely the opposite is true in the $\mathcal{L}^-$ level expansion.

Let us make a technical comment: Fields in the subalgebra of 
wedge states with insertions are not necessarily analytic functions of $K$ 
around $K=\infty$. This is why noninteger powers of $K$, and the remainder 
$R(K)$, can appear in \eq{FKexp2} while (by assumption) no such terms appear
in the expansion around $K=0$. But this means that the expansion around 
$K=\infty$ can produce terms (for example $\frac{1}{\ln K}$) which are not 
$\mathcal{L}^-$ eigenstates, and are not accounted for in \eq{FKexp2}. In this
situation the field does not have a dual $\mathcal{L}^-$ 
level expansion, though much of our discussion can be extended to such 
examples. At any rate, in most applications the string field
does have a dual $\mathcal{L}^-$ level expansion. For illustrative purposes, 
we have listed the $\mathcal{L}^-$ and dual $\mathcal{L}^-$ level expansions 
of some commonly encountered states in table \ref{tab:Lm}.

\begin{table}[t]
\renewcommand{\arraystretch}{1.0}
\begin{center}
\begin{tabular}{|c|c|c|}
\hline 
& $\mathcal{L}^-$ expansion & dual $\mathcal{L}^-$ expansion \\
\hline & & \\
$\Omega$ & 
$\displaystyle 1-K+\frac{K^2}{2!} - \frac{K^3}{3!} +...$ &  
$ R(K)\ \ \ \ \ R(K)=\Omega$ \\
& &\\
\hline
& &\\
$\displaystyle\frac{1}{1+K}$ & 
$1-K+K^2-K^3+...$ & 
$\displaystyle \frac{1}{K}-\frac{1}{K^2}+\frac{1}{K^3}-...$\\
& &\\
\hline
& &\\
$1+K$ & $1+K$ & $K+1$ \\
& &\\
\hline
& &\\
$\displaystyle \frac{1}{\sqrt{1+K}}$ & 
$\displaystyle \ \ \ 1-\frac{1}{2}K + \frac{3}{8}K^2 - \frac{5}{16}K^3+...
\ \ \ $ &
$\displaystyle\ \ \ \frac{1}{K^{1/2}}-\frac{1}{2}\frac{1}{K^{3/2}}+
\frac{3}{8}\frac{1}{K^{5/2}}-\frac{5}{16}\frac{1}{K^{7/2}}+...\ \ \ $ \\
& &\\
\hline
& &\\
$\displaystyle \frac{1-\Omega}{K}$ & 
$\displaystyle 1 -\frac{K}{2!} + \frac{K^2}{3!} -\frac{K^3}{4!}+...$ & 
$\displaystyle \frac{1}{K}+R(K)\ \ \ \ R(K) = -\frac{1}{K}\Omega$\\
& &\\
\hline
& &\\
$\displaystyle \frac{K}{1-\Omega}$ & 
$\displaystyle 1 -\frac{K}{2} + \frac{K^2}{12} -\frac{K^4}{720}+...$ & 
$\displaystyle K+R(K)\ \ \ \  R(K) = \frac{K\Omega}{1-\Omega}$\\
& &\\
\hline
\end{tabular}
\end{center}
\caption{\label{tab:Lm} $\mathcal{L}^-$ and dual $\mathcal{L}^-$ level 
expansions of various states in the wedge algebra.} 
\end{table}

Now let us return to the original question, which was to find an alternative
expression for the regularity condition \eq{cond1}. Via equation \eq{cond2},
it is clear that \eq{cond1} imposes a constraint on the leading power $\nu_1$
in the expansion of $F(K)$ around $K=\infty$:
\begin{equation}h-\nu_1<0.\end{equation}
The quantity $h-\nu_1$ is the highest level in the dual $\mathcal{L}^-$ level
expansion of $F(K)\phi$. Recalling theorem \ref{thm:pre_Lm}, this 
suggests the general result:
\begin{theorem} \label{thm:Lm} 
Let $\Phi$ be a state in the subalgebra of wedge states with insertions that 
admits a dual $\mathcal{L}^-$ level expansion. Then $\Tr[\Phi]$ is well 
defined only if the highest level in the dual $\mathcal{L}^-$ level 
expansion of $\Phi$ is strictly negative.\end{theorem}
\begin{proof}A general field in the subalgebra of wedge states with insertions
can be expressed as a linear combination of states of the form
\begin{equation}\Phi = \int_0^\infty dt_1dt_2...dt_n\, f(t_1,t_2,...,t_n)\, 
\Omega^{t_1}\phi_1\Omega^{t_2}\phi_2...\Omega^{t_n}\phi_n.\end{equation}
For notational convenience we place the insertion $\phi_n$ at the right edge 
of $\Phi$. If $\Phi$ has no nontrivial insertion there, we can set $\phi_n=1$.
Taking the Laplace transform of $f(t_1,...,t_n)$ defines the function
\begin{equation}F(K_1,...,K_n) = \int_0^\infty dt_1...dt_n\, f(t_1,...,t_n)\,
e^{-t_1K_1}...e^{-t_nK_n}.\end{equation}
We can think of $F(K_1,...,K_n)$ either as a function of $n$ numbers $K_i$,
or as a function of a single string field $K$ with the understanding the index 
on $K_i$ tells us how to order $K$ relative to the insertions $\phi_i$. Taking
the trace, we can reorganize the integrals over $t$s into an integral 
over the total width of the cylinder $L = t_1+...+t_n$ and an integral over
angular parameters $\theta_i = t_i/L$ separating the insertions:
\begin{eqnarray}\Tr[\Phi]\lineup = \int_0^\infty dt_1...dt_n\, 
f(t_1,...,t_n)\,
\Tr[\Omega^{t_1}\phi_1...\Omega^{t_n}\phi_n]\nonumber\\
\lineup = \int_0^\infty dL\, L^{n-1}\int 
d\theta_1...d\theta_n \,\delta(\theta_1+...+\theta_n-1)
f(L\theta_1,...,L\theta_n)\,\Tr[\Omega^{L\theta_1}\phi_1...
\Omega^{L\theta_n}\phi_n].\nonumber\\ \end{eqnarray}
The identity string field appears in the $L\to 0$ limit of the integration. 
With a reparameterization we can pull the trace out of the 
integration over $L$:
\begin{equation}\Tr[\Phi]= \int 
d\theta_1...d\theta_n \,\delta(\theta_1+...+\theta_n-1)\,
\Tr[\Omega^{\theta_1}\phi_1...\Omega^{\theta_n}\phi_n]
\int_0^\infty dL\, \left(\frac{1}{L}\right)^{h_1+...+h_n+1-n}
\,f(L\theta_1,...,L\theta_n).\end{equation}
We assume that this quantity is well defined only if the integration over $L$
is convergent towards $L=0$. This requires
\begin{equation}\lim_{L\to 0} \left(\frac{1}{L}\right)^{h_1+...+h_n-n}
f(L\theta_1,...,L\theta_n) = 0.\end{equation}
This can be equivalently stated
\begin{equation}\lim_{\alpha\to\infty}\alpha^{h_1+...+h_n-n}
f\left(\frac{t_1}{\alpha},...,\frac{t_n}{\alpha}\right)=0.\end{equation}
Now consider the limit: 
\begin{eqnarray}\lim_{\alpha\to\infty} \alpha^{h_1+...+h_n}
F(\alpha K_1,...\alpha K_n)\lineup = \lim_{\alpha\to\infty} 
\alpha^{h_1+...+h_n}\int_0^\infty dt_1...dt_n\, f(t_1,...,t_n)\,
e^{-\alpha t_1 K_1}...e^{-\alpha t_nK_n}\nonumber\\
\lineup = \int_0^\infty dt_1...dt_n\,\left[\lim_{\alpha\to\infty} \alpha^{h_1+...+h_n-n}f\left(\frac{t_1}{\alpha},...,\frac{t_n}{\alpha}\right)\right]
e^{- t_1 K_1}...e^{- t_nK_n}\nonumber\\
\lineup = 0.
\end{eqnarray} This means that $F(K_1,...,K_n)$ vanishes at $K_i\to\infty$ 
faster the inverse power of $h_1+...+h_n$. Therefore, if $\Phi$ admits
a dual $\mathcal{L}^-$ level expansion, the leading level must be strictly 
negative.
\end{proof}
\noindent One corollary of this result is that solutions 
cannot have zero or positive levels in the dual $\mathcal{L}^-$ level 
expansion, since the action and closed string overlap \cite{Ellwood} 
must be well-defined. Usually it is easy to check this by 
reading off the leading behavior of the solution around $K=\infty$. 
For example, let's look at the ``simple'' tachyon vacuum \cite{simple} and the 
KOS marginal solution \cite{bcc}:
\begin{eqnarray}
\Psi_\mathrm{simple} \lineup = (c+Q(Bc))\frac{1}{1+K}, \nonumber\\
\Psi_\mathrm{KOS} \lineup = \frac{1}{\sqrt{1+K}}Q\sigma_{01}\frac{B}{1+K}\sigma_{10}(1+K)c
\frac{1}{\sqrt{1+K}}.
\end{eqnarray}
Expanding around $K=\infty$ gives the leading behavior:
\begin{eqnarray}
\Psi_\mathrm{simple} \lineup = Q(Bc)\frac{1}{K} + ...\ \ ,
 \label{eq:lead_simp}\\
\Psi_\mathrm{KOS} \lineup = \frac{1}{\sqrt{K}}\left[Q\sigma_{01}\frac{B}{K}
Q\sigma_{10}
+\sigma_{01}Q\sigma_{10}\right]\frac{1}{\sqrt{K}}+...\ \ .
\end{eqnarray}
Recalling that $c$ has scaling dimension $-1$ and the boundary condition 
changing operators $\sigma_{01}$ and $\sigma_{10}$ have dimension $0$, we can 
see that the leading level for both solutions is $-1$. This is the 
highest possible (integer) level consistent with a regularity, so in a 
sense these solutions are as identity-like as possible. 
There are also well-known identity-based 
solutions for the tachyon vacuum \cite{id1,id2} and marginal deformations 
\cite{rolling}
\begin{equation}\Psi_\mathrm{tv} = c(1-K),\ \ \ \ \ \ \Psi_\mathrm{marg} 
= cV.\end{equation}
Since the marginal operator $V$ has dimension $1$, the leading level for both
of these solutions is level $0$, which means that they are too identity-like.

Let us mention one other result, which is important for the physical 
interpretation of the dual $\mathcal{L}^-$ level expansion:
\begin{theorem} \label{thm:gauge}
Let $\Psi$ be a regular solution which admits a dual 
$\mathcal{L}^-$ level expansion. Then all finite levels in the 
dual $\mathcal{L}^-$ level expansion of $\Psi$ can be eliminated by a gauge 
transformation.
\end{theorem}
\begin{proof}The dual $\mathcal{L}^-$ level expansion of $\Psi$ takes 
the form
\begin{equation}\Psi = \Psi_h + \mathrm{lower\ levels},\end{equation}
where $\half\mathcal{L}^-\Psi_h = h\Psi_h$ and ``lower levels'' denotes a 
sequence of $\mathcal{L}^-$ eigenstates with eigenvalue less than $h$. Since
by assumption $\Psi$ is a well-defined solution, theorem \ref{thm:Lm} implies
that $h$ is less than zero. Now consider the level $h$ contribution to 
the equations of motion $Q\Psi+\Psi^2=0$. Since $h$ is negative, the 
highest level of $\Psi^2$ is necessarily less than $h$, so the quadratic
term does not contribute at this level. Therefore the equations of motion imply
\begin{equation}Q\Psi_h = 0.\end{equation}
Now we want to explicitly remove the term $\Psi_h$ by a gauge transformation.
Assume the gauge transformation takes the form
\begin{equation}U=1+U_h +\mathrm{lower\ levels}.\end{equation}
Thus we require that $U^{-1}(Q+\Psi)U$ vanishes at level $h$, 
which implies
\begin{equation}QU_h +\Psi_h=0.\end{equation}
We can solve this equation only if $\Psi_h$ is BRST exact. However, at finite
levels all BRST closed states are exact, since we can define the 
homotopy operator
\begin{equation}A=\frac{B}{K}.\end{equation}
Of course this is a singular state, but according to the current 
philosophy this is not a problem. The sliver divergence can be 
arranged to cancel against lower level states. Therefore we can remove 
the leading level term $\Psi_h$ with a gauge transformation of the form
\begin{equation}U = 1-\frac{B}{K}\Psi_h+\mathrm{lower\ levels}.\label{eq:rem}
\end{equation}
We assume that the lower level states can be chosen so that $U$ is well-defined
and invertible. This means that the highest level of $\Psi$ can be eliminated 
by a gauge transformation, and by repeating this process iteratively, 
all finite levels can be eliminated.
\end{proof}
\noindent This result means that there is no gauge invariant information 
revealed by the dual $\mathcal{L}^-$ level expansion. In a sense, 
there is no ``physics'' in the identity string field. All of the physics of 
is carried by states at level $-\infty$. Note that the prototypical 
example of a level $-\infty$ state is the sliver; Therefore, theorem 
\ref{thm:gauge} fits well with 
studies of the phantom term \cite{Schnabl,phantom} and vacuum string field 
theory \cite{VSFT,RSZsol,RSZproj}, which suggest that all of the physics of 
open string field theory can be encoded in finite rank projectors.

Let illustrate the method behind the proof of theorem \ref{thm:gauge} with 
an example. Suppose that for some application the ``simple''
tachyon vacuum solution is too identity-like, and we want to remove 
the leading term in the dual $\mathcal{L}^-$ level expansion. 
The leading term is written in \eq{lead_simp}. Plugging
into \eq{rem} we find that the required gauge transformation takes the form
\begin{equation}U = 
1-\frac{B}{K}Q(Bc)\frac{1}{K}+...=1-Bc\frac{1}{K}+...\ \ .
\end{equation}
To complete the definition of $U$ we should fix the lower level states 
according to our convenience. The first priority is to get rid of the pole at 
$K=0$. This can be achieved, for example, by replacing $\frac{1}{K}$ with 
$\frac{1}{K+1}$, but unfortunately the 
resulting gauge transformation is not invertible. 
So instead we can replace $\frac{1}{K}$ with $\frac{1}{K+2}$:
\begin{equation}U = 1-Bc\frac{1}{K+2}.\end{equation}
Transforming the ``simple'' tachyon vacuum with this gauge parameter happens 
to give the solution
\begin{eqnarray}\Psi' \lineup =
U^{-1}(Q+\Psi_\mathrm{simple})U \nonumber\\
\lineup = \frac{K+2}{K+1}\left[c\frac{B}{K+2}c+Q(Bc)\right]
\frac{1}{(K+1)(K+2)}.
\end{eqnarray}
The dual $\mathcal{L}^-$ level expansion takes the form
\begin{equation}\Psi' = Q(Bc)\frac{1}{K^2}+...\ \ .
\end{equation}
Now the leading level is $-2$. The level $-1$ state of the ``simple'' tachyon 
vacuum has been removed as desired.

\section{Applications to Solutions in the $KBc$ subalgebra}

In this section we turn our attention to solutions in the $KBc$ subalgebra. 
The string fields $K$ and $B$ were discussed before. The field $c$ corresponds
to a local insertion of the $c$ ghost in correlation functions on the 
cylinder. It satisfies the identities
\begin{equation}c^2=0,\ \ \ \ \ \ \ \ cB+Bc=1,\ \ \ \ \ \ \ \ 
Qc=cKc,\label{eq:c_id}
\end{equation}
and has scaling dimension $-1$:
\begin{equation}\half\mathcal{L}^- c = -c.\end{equation}
Equations \eq{c_id} and \eq{KB_id} define what we call the 
{\it $KBc$ subalgebra}. Note that $c$ satisfies additional relations 
which are not implied by the $KBc$ subalgebra, for example 
$(\d c)^2=0$ (using the notation $\d \equiv [K,\,\cdot\,]$). We call 
these {\it auxiliary identities}.\footnote{The 
general set of auxiliary identities is $(\d^m c)^2=0$ for all $m\geq 1$. 
This implies $\d^m c \,\d^n c = - \d^n c\,\d^m c$. 
Equation \eq{c_id} only implies the cases $m=n=0$ and $m=0,n=1$.} A general
realization of the $KBc$ subalgebra does not satisfy these identities. 
See, for example, equation \eq{aut} later.
To simplify the discussion we will focus on solutions which satisfy the 
equations of motion by virtue of the basic 
relations of the $KBc$ subalgebra, \eq{KB_id} and \eq{c_id}, alone. 
In particular, theorems \ref{thm:comp} and \ref{thm:nogo} apply to solutions 
in this class. Auxiliary identities produce further solutions which we have 
not systematically analyzed, but play a role in later discussion. 

Our task is to identify all possible gauge orbits of solutions in the 
$KBc$ subalgebra using the $\mathcal{L}^-$ level expansion, and then 
to investigate the regularity of these solutions using the dual 
$\mathcal{L}^-$ level expansion. Our results can be summarized as follows: 
We identify $6$ gauge equivalence classes of solutions:
\begin{enumerate}
\item[1)] Perturbative vacuum

\item[2)] Tachyon vacuum 

\item[3)] Residual perturbative vacuum

\item[4)] Residual tachyon vacuum

\item[5)] Residual conjugate tachyon vacuum

\item[6)] MNT ghost brane
\end{enumerate}
The last four types of solution are unexpected. We will 
call them {\it residual solutions}.\footnote{The residual 
tachyon vacuum solutions are 
discussed by Zeze \cite{id2}. The residual perturbative 
vacuum solutions were first pointed out to the author by C. Maccaferri, and 
they were also noted by MNT \cite{MNT}.} As we will see, they are 
singular from the perspective of the identity string field. 
In particular, residual solutions always have zero or positive levels in 
the dual $\mathcal{L}^-$ level expansion. 

Consider the $\mathcal{L}^-$ level expansion of a solution $\Psi$ and 
gauge parameter $U$ in the $KBc$ subalgebra:
\begin{eqnarray}\Psi = \lineup \Psi_{-1}+\Psi_0+\Psi_1+\Psi_2+...,\nonumber\\
\lineup \Psi_{-1} \equiv \alpha\, c,\nonumber\\
\lineup \Psi_0 \equiv  \gamma_1\, cK +\gamma_2\, Kc + \beta\, cKBc,\nonumber\\
\lineup \ \ \ \vdots\nonumber\\
U= \lineup U_0+U_1+U_2+...,\nonumber\\
\lineup U_0 \equiv 1+\lambda Bc\ \ \ \ (\lambda\neq -1), \nonumber\\
\lineup \ \ \ \vdots 
\label{eq:PsinUn}
\end{eqnarray}
The index on the eigenstates $\Psi_n$ and $U_n$ refers to their 
$\half\mathcal{L}^-$ eigenvalue, and the constants 
$\alpha,\beta,\gamma_1,\gamma_2$ are coefficients to be 
determined by solving the equations of motion. The constant $\lambda$
in $U$ can take any value besides $-1$, which is not allowed
since $U$ must be invertible. The equations of motion imply
\begin{equation}\Psi_{-1}^2 =0,\end{equation}
which is satisfied for any choice of the coefficient $\alpha$ in front of 
$c$. Transforming $\Psi$ with $U$, it is easy to show that all nonzero 
choices of $\alpha$ can be related by a gauge transformation with the 
appropriate choice of $\lambda$. However, $\alpha =0$ and $\alpha\neq 0$ 
cannot be related by a gauge transformation. Solutions with $\alpha\neq 0$ 
turn out to describe the tachyon vacuum: 
\begin{equation}\mathrm{Tachyon\ vacuum}:\ \ \ 
\Psi = \alpha\, c +\,...\ \ \ \ (\alpha\neq 0),
\label{eq:m1gauge}\end{equation}
where $...$ denotes higher level terms. 
Now let's look at the $\alpha= 0$ solutions. The $\mathcal{L}^-$ 
level expansion now takes the form
\begin{equation}\Psi = \Psi_0+\Psi_1+\Psi_2+...\ \ .\end{equation}
The equations of motion imply that the leading contribution $\Psi_0$ 
itself satisfies the equations of motion: 
\begin{equation}Q\Psi_0+\Psi_0^2 = 0.\end{equation}
Plugging in $\Psi_0$ from \eq{PsinUn}, straightforward algebra
reveals five possible solutions:
\begin{eqnarray}
\lineup \mathrm{Perturbative\ vacuum}:\ \ \ \ \ \ \ \ \ \ \ 
\ \ \ \ \ \ \ \ \ \ \ \ \ 
\Psi=\beta\,cKBc +\,...,
\ \ \ \ \ (\beta\neq -1) \nonumber\\
\lineup \mathrm{Residual\ perturbative\ vacuum}:\ \ \ \ \ \ \ \ \ \ \ \ 
\Psi = -cKBc
+\,...,\nonumber\\
\lineup \mathrm{Residual\ tachyon\ vacuum}:\ \ \ \ \ \ \ \ \ \ \ \ \ \ \ \ \ 
\  
\Psi = -cK
+\,...,\nonumber\\
\lineup \mathrm{Residual\ conjugate\ tachyon\ vacuum}:\ \ \ \ \,
\Psi = -Kc
+\,...\nonumber\\
\lineup \mathrm{MNT\ ghost\ brane}:\ \ \ \ \ \ \ \ \ \ \ \ \ \ 
\ \ \ \ \ \ \ \ \ \ \ \ \ \ 
\Psi = -cK-Kc+cKBc+
\,...\ \ .
\label{eq:0gauge}\end{eqnarray}
The last four solutions---the residual solutions---are invariant under gauge 
transformations by $U$ at this level. A gauge transformation of the 
perturbative vacuum can set the coefficient $\beta$ to any value 
(with the appropriate choice of $\lambda$), 
{\it except} $\beta=-1$. Note that the two residual tachyon vacuum solutions
are not gauge equivalent, though they are related by conjugation.

Thus we have extracted a total of six gauge orbits. Now we can ask whether
analysis of higher level states will reveal further physically distinct 
solutions. The answer is no, according to the following theorem:
\begin{theorem}\label{thm:comp}
Equations \eq{m1gauge} and \eq{0gauge} are the only gauge orbits for solutions 
in the $\mathcal{L}^-$ level expansion of the $KBc$ subalgebra.
\end{theorem}
\noindent The proof of this statement is somewhat lengthy so we 
postpone it to appendix \ref{app:thm}. Note that this theorem only applies 
to solutions as they are defined in the $\mathcal{L}^-$ 
level expansion. Solutions which are not analytic at $K=0$ 
(for example, multibranes and fractional branes \cite{multiple}) do not fall 
within this classification. Moreover, the $\mathcal{L}^-$ level expansion 
by itself is not a complete definition of the string field, which should 
satisfy many regularity conditions which cannot be seen before 
resummation. Such considerations may exclude certain gauge orbits 
(for example, as we will argue, the residual solutions) or it can reveal 
further gauge orbits within the ones already described. The only example 
of the later phenomenon known to us appears in the so-called half-brane 
solutions of cubic superstring field theory \cite{exotic}. We are not aware
of such phenomena in the $KBc$ subalgebra.

It is quite surprising to find four distinct solutions 
in addition to perturbative vacuum and the tachyon vacuum. What do these 
solutions mean?  One basic quantity we can calculate is the energy. 
Actually, there is a rule of thumb which says that the energy 
(relative to the perturbative vacuum) is the sum of the
coefficients of $cK$ and $Kc$ in the $\mathcal{L}^-$ level expansion:
\begin{equation}\mathrm{Energy}\sim \gamma_1+\gamma_2.\label{eq:g1g2}
\end{equation}
Thus the residual perturbative vacuum has zero energy, the residual 
tachyon vacuum has minus the energy of the reference D-brane (i.e.
the same energy as the tachyon vacuum), and the MNT ghost brane
has minus twice the energy of the reference D-brane. To see where the 
rule \eq{g1g2} comes from, first note that $\gamma_1+\gamma_2$ is a 
gauge invariant quantity for all ghost number $1$ states in the $KBc$ 
subalgebra. Therefore, if $\gamma_1+\gamma_2$ computes the energy for one 
solution in each gauge orbit, it computes the energy for 
all solutions. So let us take a representative solution from each gauge
orbit which terminates at level $0$ in the $\mathcal{L}^-$ level 
expansion:
\begin{equation}\Psi = \alpha\, c+\gamma_1\, cK +\gamma_2\,Kc+\beta\, cKBc,
\label{eq:0trunc}\end{equation}
where the coefficients are fixed according to which gauge orbit we 
are describing. We cannot compute the energy of \eq{0trunc} directly since 
the solution is too identity-like. To fix this problem, 
we regularize the $KBc$ subalgebra by defining the fields
\cite{simple_talk}
 \begin{eqnarray}\hat{K} \lineup \equiv \frac{K}{1+\eps K},\nonumber\\
\hat{B}\lineup \equiv \frac{B}{1+\eps K},\nonumber\\
\hat{c}\lineup \equiv c(1+\eps K)B c.\label{eq:aut}\end{eqnarray}
It is easy to check that $\hat{K},\hat{B}$ and $\hat{c}$ satisfy the 
defining relations \eq{c_id} and \eq{KB_id} of the $KBc$ 
subalgebra.\footnote{We can generalize \eq{aut} given any $F(K)$ satisfying
$F(0)=1$ by
\begin{equation}\hat{K} = 1-F,\ \ \ \ \ \hat{B} = B\frac{1-F}{K},\ \ \ \ \ 
\hat{c} = c\frac{KB}{1-F}c\end{equation}} Moreover, the redefined fields
$\hat{K},\hat{B}$ are more regular than $K,B$ in the dual $\mathcal{L}^-$ 
level expansion. With this replacement the solution becomes
\begin{eqnarray}\hat{\Psi} \lineup 
= \alpha\, \hat{c} + \gamma_1\,\hat{c}\hat{K}
+\gamma_2\,\hat{K}\hat{c}+\beta\,\hat{c}\hat{K}\hat{B}\hat{c}\nonumber\\
\lineup = -\gamma_1\left(\frac{1}{\eps}c +Q(Bc)\right)\frac{1}{1+\eps K}
-\gamma_2\, \frac{1}{1+\eps K}\left(\frac{1}{\eps}c +Q(Bc)\right)
\nonumber\\ \lineup \ \ \ \ \ \ \ \ \ \ \ \ \ \ \ \ \ \ \ \ +\left(\alpha+
\frac{\gamma_1+\gamma_2}{\eps}\right)c+(\alpha+\gamma_1+\gamma_2+\beta)cKBc.
\label{eq:hat0trunc}\end{eqnarray}
Now we are in a better position to compute the energy. We could compute
the action, but it is a little easier to compute the closed string overlap 
$\Tr_\mathcal{V}[\hat{\Psi}]$, which for our purposes is equivalent.
The first two terms in \eq{hat0trunc} are proportional to 
(reparameterizations of) of the ``simple'' tachyon vacuum solution 
\cite{simple}, and therefore contribute proportionally 
to the closed string overlap of the tachyon vacuum. The 
third term does not contribute to the overlap because $c$ has negative 
dimension. The fourth term is problematic since it is a level $0$ state in 
the dual $\mathcal{L}^-$ level expansion, and shouldn't have a well-defined
trace. However, since $cKBc$ is BRST exact it is natural to 
assume 
\begin{equation}\Tr_\mathcal{V}[cKBc]\equiv 0.\label{eq:ass}\end{equation} 
Under this assumption the closed string overlap is 
\begin{equation}\Tr_\mathcal{V}[\hat{\Psi}] = -(\gamma_1+\gamma_2)\,\times\,
\mathrm{tachyon\ vacuum\ overlap},\end{equation}
and therefore the energy is proportional to $\gamma_1+\gamma_2$, as claimed.
This argument is more rigorous for pure gauge and tachyon vacuum solutions,
since there it is possible to find representatives of the gauge orbit where
$\alpha+\gamma_1+\gamma_2+\beta$ vanishes, and we don't need
to assume anything about the trace of $cKBc$. For residual solutions in
the form \eq{hat0trunc} the state $cKBc$ is present. In a moment 
we will show that all residual solutions have such an identity-like term. 

Let us mention a curiosity which raises doubt as to whether 
residual solutions have well-defined energy. When we computed the 
solutions in \eq{0gauge}, we
assumed that $K,B$ and $c$ satisfy only the basic identities \eq{KB_id}
and \eq{c_id}. However, if we also account for the auxiliary identity
$(\d c)^2=0$, the solution space at level zero is enlarged. The 
four residual solutions become special cases of a two 
parameter family of solutions\footnote{The first example 
of a solution utilizing an auxiliary identity was pointed out to the 
author by M. Schnabl.}
\begin{equation}\Psi_0 = \gamma_1\, cK +\gamma_2\,Kc - (1+\gamma_1+\gamma_2)
cKBc ,\label{eq:aux}\end{equation}
for arbitrary $\gamma_1$ and $\gamma_2$. This means that all four residual 
solutions are related by marginal deformations, where $\gamma_1$ and 
$\gamma_2$ are marginal parameters  (see figure \ref{fig:IdSing2}). 
Therefore residual solutions must have the same energy, 
in contradiction with our earlier reasoning. It would be interesting to see 
how this observation can be reconciled with MNT's calculation of the boundary 
state \cite{MNT,boundary}.

\begin{figure}
\begin{center}
\resizebox{2.7in}{2in}{\includegraphics{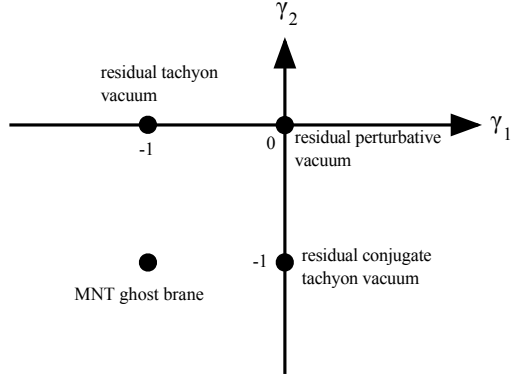}}
\end{center}
\caption{\label{fig:IdSing2} Moduli space of residual solutions. The 
horizontal and vertical axis represent expectation values of the marginal
fields $\gamma_1$ and $\gamma_2$ from \eq{aux}. The four 
points are the four residual solutions \eq{0gauge} whose existence depends only
on the basic algebraic relations of the $KBc$ subalgebra. The remaining 
solutions require the additional relation $(\d c)^2=0$ to satisfy the 
equations of motion.}\end{figure}

These problems appear to be related to the fact that residual solutions 
are too identity-like. However, we have not shown that all residual solutions
suffer from this problem. This is the purpose of the following theorem:
\begin{theorem}\label{thm:nogo}
Let $\Psi$ be a residual solution in the $KBc$ subalgebra. 
Then the highest level in the dual $\mathcal{L}^-$ level expansion of $\Psi$
is zero or positive.
\end{theorem}
\begin{proof} A solution in the $KBc$ subalgebra takes the general form
\begin{equation}\Psi = \int_0^\infty dt_1dt_2dt_3\, f(t_1,t_2,t_3)\,
\Omega^{t_1}c\Omega^{t_2}Bc\Omega^{t_3},\end{equation}
where the function $f(t_1,t_2,t_3)$ is determined by solving the equations 
of motion. Taking the Laplace transform of $f(t_1,t_2,t_3)$ 
gives the function
\begin{equation}F(K_1,K_2,K_3)=\int_0^\infty dt_1dt_2dt_3\,f(t_1,t_2,t_3)
e^{-t_1K_1}e^{-t_2K_2}e^{-t_3K_3}.\end{equation}
To make things easier to read we will sometimes omit the $K$ 
when it appears in the argument of $F$; for example, 
$F(1,2,3)\equiv F(K_1,K_2,K_3)$. With a little algebra we can show that the 
equations of motion imply a functional equation for $F$: 
\begin{equation}K_2\, F(1,3,4)-K_3\,F(1,2,4)+F(1,2,2)F(2,3,4)
-F(1,2,3)F(3,3,4)=0.\label{eq:KEOM}
\end{equation}
This equation depends on four variables $K_1,K_2,K_3,K_4$. Setting 
$K_1=K_2$ and $K_3=K_4$ we find
\begin{equation}\Big[K_2+F(2,2,2)\Big]F(2,3,3) = \Big[K_3+F(3,3,3)\Big]
F(2,2,3).\label{eq:constr}\end{equation}
Let's see what this implies about residual solutions. From \eq{0gauge} 
we find that all four residual solutions have the property 
\begin{equation}F(K,K,K) = -K+...,\label{eq:resconstr}\end{equation}
where $...$ denotes possible higher powers of $K$. 
This implies, for example,
\begin{equation}K_2+F(2,2,2)= p\, K_2^n+...\ \ \ \ \ 
(n\geq 2),
\end{equation}
where $p$ is some constant and $n$ is any power greater than $1$. Note that 
\eq{resconstr} also determines the behavior of $F(2,3,3)$ and $F(2,2,3)$
near $K=0$:
\begin{eqnarray}F(2,3,3)\lineup =a\,K_2 - (a+1)\,K_3+...\nonumber\\
F(2,2,3)\lineup = b\,K_2 -(b+1)\,K_3+...,\end{eqnarray}
where $a,b$ are constants. Plugging these into \eq{constr}, and focusing 
on the leading term in $K$, gives a constraint of the 
constants $p,a,b$ from the equations of motion:
\begin{equation}p\, K_2^n(a\, K_2 - (1+a)K_3) = p\,K_3^n(b\, K_2 -(1-b)K_3).
\end{equation}
When $n\geq 2$, the only solution to this equation is $p=0$. Therefore 
all higher order corrections in \eq{resconstr} vanish, and  
we have the exact equality:
\begin{equation}F(K,K,K)= -K.\end{equation}
This means, in particular, that {\it all} residual solutions in the $KBc$
subalgebra satisfy the identity
\begin{equation}B\Psi B = -BK.\end{equation}
The right hand side is a level $2$ state in the dual $\mathcal{L}^-$ level 
expansion. Since the highest level of a product of states is less than or 
equal to the sum of the highest levels of the states individually, this means
that the highest level in the dual $\mathcal{L}^-$ level expansion of any
residual solution is zero or positive. 
\end{proof}

As a final comment, let us mention an oddity related to the 
characteristic projector \cite{Ian} of residual solutions. The characteristic
projector is the boundary condition changing projector of a singular gauge 
transformation from a solution to itself \cite{Integra}. It is expected to 
give information about the boundary conformal field theory corresponding 
to a classical solution. The simplest example of a singular gauge 
transformation from a residual solution to itself takes the form
\begin{equation}U = Q_{\Psi_0} B = -(1+\gamma_1+\gamma_2)B\d c,
\label{eq:sing}\end{equation}
where we take $\Psi_0$ from \eq{aux}. Surprisingly, this is exactly the 
type of singular gauge transformation for which the boundary condition 
changing projector does not exist. In particular, $U$ is nilpotent, so its 
kernel and image are not linearly independent and do not define the image 
and kernel of a projector. One possible interpretation of this result is 
that the projector does not exist since the solutions do not describe 
a boundary conformal field theory. In fact, all four residual solutions 
have physical cohomology in the universal sector 
(see appendix \ref{app:thm}), which is difficult to reconcile with a 
boundary conformal field theory interpretation.

\section{Discussion}

In conclusion, let us discuss the implications of our analysis for the MNT 
ghost brane solutions. We have shown that these solutions are necessarily 
singular from the perspective of the identity string field. 
In the best case scenario, they have an identity-like term of the form
\begin{equation}cKBc.\label{eq:cKBc}\end{equation}
We claim that such a term renders the action undefined. This requires 
a little explanation. In particular \eq{cKBc} contributes to (for example) 
the cubic term in the action as
\begin{equation}\Tr[(cKBc)^3].\label{eq:trcKBc}\end{equation}
This appears to vanish without ambiguity both because $cKBc$ is 
BRST exact and because $(cKBc)^3$ is a vanishing state. 
Nevertheless, a generic regularization of \eq{trcKBc} does not vanish. 
Consider for example the regularization
\begin{eqnarray}\Tr[(cKBc)^3] \lineup 
= \lim_{\eps\to 0}\Tr\left[
\Big(cKB\,\Omega^{\eps t_1}\,c\,\Omega^{\eps t_2}\Big)^3\right]\nonumber\\
\lineup = \Tr\left[
\Big( cKB\,\Omega^{t_1}\,c\,\Omega^{ t_2}\Big)^3\right].
\end{eqnarray}
This quantity vanishes if $t_1=0$ because $cKBc$ is BRST exact, and it also 
vanishes if $t_2=0$ because $c^2 = 0$, but it does not vanish for generic 
values of $t_1,t_2$.

Therefore if we want to define the MNT solution we need to apply some 
regularization. We can do this, for example, by replacing
$cKBc$ with 
\begin{equation}cKBc \ \rightarrow\  cKBc\frac{1}{1+\eps K}. \label{eq:0reg}
\end{equation}
and taking the $\eps\to 0$ limit. This approach appears to be consistent.
The equations of motion are satisfied in the Fock space and when contracted 
with the solution. The situation might be contrasted to solutions 
with sliver-like singularities \cite{multiple,lumps}, where 
regularization typically produces problems with the equations of motion.
What makes the identity-like singularities of MNT more ``mild'' is 
that (what might be called) the ``dual'' sliver state,
\begin{equation}(1-\Omega)^\infty,\end{equation}
vanishes in the Fock space, whereas the sliver state does not. Nevertheless,
it seems that the MNT ghost brane and related solutions are not completely 
healthy. Perhaps further study will clarify whether a regularization
such as \eq{0reg} truly defines an acceptable solution.

\bigskip

\noindent {\bf Acknowledgments}

\bigskip

\noindent I would like to thank C. Maccaferri for looking over a draft 
of the paper, and I. Sachs for kind hospitality in Munich where much 
of this work was completed. This research was supported by the 
Grant Agency of the Czech Republic under the grant P201/12/G028.

\begin{appendix}
\section{Proof of Theorem \ref{thm:comp}}
\label{app:thm}
In this appendix we prove theorem \ref{thm:comp}, which says
that any pair of solutions in the $KBc$ subalgebra 
which share the same leading term in the $\mathcal{L}^-$ level 
expansion can be related by a gauge transformation. 
We will prove this by induction; We show that if two solutions are 
equal up to level $n$, a gauge transformation can make them equal up to 
level $n+1$. Thus the solutions are gauge equivalent.

Let's start with tachyon vacuum solutions. Suppose we are given 
two tachyon vacuum solutions which are equal up to level $n$ in the 
$\mathcal{L}^-$ level expansion, but differ at level $n+1$:
\begin{eqnarray}
\Psi \lineup = (\Psi_{-1}+\Psi_0+...+\Psi_n)
+\Psi_{n+1}+...,\nonumber\\ 
\Psi'\lineup = (\Psi_{-1}+\Psi_0+...+\Psi_n)+\Psi_{n+1}'+...\ \ .
\end{eqnarray}
Our task is to construct a gauge transformation which will make these solutions
identical up to level $n+1$. Assume that the gauge parameter takes the form
\begin{equation}U = 1+U_{n+2}+...,\end{equation}
where $U_{n+2}$ is a level $n+2$ state and $...$ denotes higher level terms. 
We fix $U_{n+2}$ by requiring that 
\begin{equation}U^{-1}(Q+\Psi')U\label{eq:ppp}\end{equation}
is equal to $\Psi$ up to level $n+1$. It is automatically equal to $\Psi$
up to level $n$ because of the assumed form of $\Psi'$ and $U$. Imposing 
equality at level $n+1$ gives the equation
\begin{equation}\Psi_{n+1}'-\Psi_{n+1}=[\Psi_{-1},U_{n+2}].\label{eq:Ueq}
\end{equation}
Note that $\Psi_{-1}$ is nilpotent because it is proportional to $c$. 
Therefore acting $[\Psi_{-1},\cdot]$ on \eq{Ueq} implies a condition
on $\Psi_{n+1}'-\Psi_{n+1}$: 
\begin{equation}[\Psi_{-1},\Psi_{n+1}'-\Psi_{n+1}]=0.\label{eq:int}
\end{equation}
This condition is implied by the equations of motion, 
and therefore does not need to be separately assumed. Next we introduce 
a string field $A_1$ at level $1$ which satisfies
\begin{equation}[\Psi_{-1},A_1]=1.\label{eq:hom}\end{equation}
If $\Psi_{-1}=\alpha c$, then we can take $A_1=\frac{1}{\alpha}B$. Taking
\eq{int} and \eq{hom} together implies that 
\begin{equation}U_{n+2} = A_1(\Psi_{n+1}'-\Psi_{n+1})\end{equation}
satisfies \eq{Ueq}. Therefore all tachyon vacuum solutions, as defined by 
the leading term in the $\mathcal{L}^-$ level expansion, are 
gauge equivalent. Note that this argument does not depend in an essential
way on the $KBc$ subalgebra. The higher levels can in principle be
composed of more complicated states outside the algebra. All we need is the 
leading level $\Psi_{-1}$ paired with an operator $A_1$ satisfying \eq{hom}.

Now let's consider the five solutions which start at level $0$. Assume
that a pair of such solutions are equal up to level $n$:
\begin{eqnarray}
\Psi \lineup = (\Psi_0+...+\Psi_n)
+\Psi_{n+1}+...,\nonumber\\ 
\Psi'\lineup = (\Psi_0+...+\Psi_n)+\Psi_{n+1}'+...\ \ .
\end{eqnarray}
We want to find a gauge transformation which makes these solutions identical
up to level $n+1$. Assume that the gauge parameter takes the form
\begin{equation}U=1+U_{n+1}+...,\end{equation}
where $U_{n+1}$ is a level $n+1$ state and $...$ denotes higher level terms.
We fix $U_{n+1}$ by requiring that 
\begin{equation}U^{-1}(Q+\Psi')U\end{equation}
is equal to $\Psi$ up to level $n+1$. Equality up to level $n$ follows from
the form of $U$ and $\Psi'$. Imposing equality at level $n+1$ gives the 
equation
\begin{equation}Q_{\Psi_0}U_{n+1} + \Psi_{n+1}'-\Psi_{n+1}=0,\label{eq:Ueq2}
\end{equation}
where $Q_{\Psi_0}$ is the kinetic operator around the level $0$ solution 
$\Psi_0$. Note that $\Psi_{n+1}'-\Psi_{n+1}$ is $Q_{\Psi_0}$-closed 
as a consequence of the equations of motion. We need to show that it
is also $Q_{\Psi_0}$-exact, so that \eq{Ueq2} has a solution for $U_{n+1}$.

To prove this, we use a variant of the standard argument that the 
cohomology of $Q$ is in the kernel of $L_0$. To start, consider the BPZ odd 
component of the $b$-ghost zero mode $\mathcal{B}_0$ 
in the sliver coordinate frame, which we call $\mathcal{B}^-$:
\begin{equation}\mathcal{B}^- \equiv 
\mathcal{B}_0-\mathcal{B}_0^\star.\end{equation}
This is a derivation of the star product and satisfies
\begin{equation}\half\mathcal{B}^- K = B,\ \ \ \ \half\mathcal{B}^- B =0,
 \ \ \ \ \half\mathcal{B}^- c =0.\end{equation}
Also 
\begin{equation}[Q,\mathcal{B}^-] = \mathcal{L}^- .\end{equation}
Now define the operator
\begin{equation}\mathcal{L}^-_{\Psi_0} \equiv [Q_{\Psi_0},\mathcal{B}^-]
=\mathcal{L}^-+\left[\big(\mathcal{B^-}\Psi_0\big),\,\cdot\,\right].
\end{equation}
If $\mathcal{L}^-_{\Psi_0}$ is diagonalizable, then the cohomology
of $Q_{\Psi_0}$ can be found in its kernel. In particular if $\phi$ is a
$Q_{\Psi_0}$ closed eigenstate of $\mathcal{L}^-_{\Psi_0}$ with eigenvalue
$h\neq 0$, then
\begin{equation}\phi =\frac{1}{h}
Q_{\Psi_0}(\mathcal{B}^-\phi),\label{eq:ex}\end{equation}
so $\phi$ is exact.

So our goal is to identify the states in the kernel of 
$\mathcal{L}^-_{\Psi_0}$ for positive levels in the $KBc$ subalgebra, 
and show that they do not lead to cohomology at ghost number $1$. 
For the five solutions at level $0$ in \eq{0gauge}, 
the operator $\mathcal{L}^-_{\Psi_0}$ takes the form
\begin{eqnarray}
\lineup \mathrm{Perturbative\ vacuum:}\ \ \ \ \  \ \ \ \ \ \ \ \
\ \ \ \ \  \ \ \ \ \ \ \ \ \ \ \ \ \  \ 
\half \mathcal{L}^-_{\Psi_0} = \half \mathcal{L}^-,
\nonumber\\
\lineup \mathrm{Residual\ perturbative\ vacuum:}
\ \ \ \ \  \ \ \ \ \ \ \ \ \ \ \ \ \  \ \ 
\half \mathcal{L}^-_{\Psi_0} = \half\mathcal{L}^-,
\nonumber\\
\lineup \mathrm{Residual\ tachyon\ vacuum:}
\ \ \ \ \  \ \ \ \ \ \ \ \ 
\ \ \ \ \  \ \ \ \ \ \ \ \ 
\half \mathcal{L}^-_{\Psi_0} = \half \mathcal{L}^-+[\,cB,\,\cdot\,],
\nonumber\\
\lineup \mathrm{Residual\ conjugate\ tachyon\ vacuum:}\ \ \ \ \,
\ \ \ \ \ \ \ \ 
\half \mathcal{L}^-_{\Psi_0} = \half \mathcal{L}^-+[\,cB,\,\cdot\,],
\nonumber\\
\lineup \mathrm{MNT\ ghost\ brane:}
\ \ \ \ \  \ \ \ \ \ \ \ \ \ \ \ \ \  \ \ \ \ \ \ \ \ 
\ \ \ \ \  \ \ \ \ \ 
\half \mathcal{L}^-_{\Psi_0} = \half \mathcal{L}^-+2[\,cB,\,\cdot\,].
\end{eqnarray}
For the perturbative vacuum and the residual perturbative vacuum 
$\mathcal{L}^-_{\Psi_0}$ simply computes twice the level, and therefore has 
no kernel if the level is positive.  
For the remaining solutions we need to diagonalize $\mathcal{L}_{\Psi_0}^-$
to see what happens. To do this, 
note that $\mathcal{L}^-$ and $[\,cB,\cdot]$ commute and so are simultaneously 
diagonalizable. A standard 
basis of $\mathcal{L}^-$ eigenstates at level $n$ and ghost number $1$ 
in the $KBc$ subalgebra is
\begin{equation}K^p\, c\, K^q \,Bc\, K^r,\ \ \ \ \ (p+q+r=n+1).\end{equation}
However, these are not eigenstates of $[\,cB,\cdot]$. To diagonalize 
$[\,cB,\cdot]$ we use an alternative basis of the states at level $n$:
\begin{eqnarray}\lineup[\,cB,\,\cdot\,]=1:\ \ \ \ \ \ \  
c\,K^{n+1}\, Bc,\\
\lineup [\,cB,\,\cdot\,]=0:\ \ \ \ \ \ \ 
c\, K^p\, B\d c\, K^q\ \ \ \ \ \ \ \ \ \ \ \ \  (p+q=n),\\
\lineup [\,cB,\,\cdot\,]=0:\ \ \ \ \ \ \ 
K^p\, \d c\, K^q\, Bc\ \ \ \ \ \ \ \ \ \ \ \ \  (p+q=n),\\
\lineup [\,cB,\,\cdot\,]=-1:\ \ \ \ \ 
K^p\, \d c\, K^q\, B \d c\, K^r\ \ \ \ \ \ \ (p+q+r=n-1).\label{eq:coh}
\end{eqnarray}
In the first three cases, the $[\,cB,\cdot]$ eigenvalue only adds to the 
level, so we don't find any states in the kernel. In the last case, however,
the $[\,cB,\cdot]$ eigenvalue subtracts, and there is the possibility of 
cohomology. For the residual tachyon vacuum solutions this occurs at level 
$1$ and for the MNT ghost brane at level $2$. At level $1$ there is a single 
state of the form \eq{coh}:
\begin{equation}B(\d c)^2.\end{equation}
This state vanishes as a consequence of the auxiliary 
identity $(\d c)^2=0$, but at any rate is trivial in the cohomology:
\begin{equation}(\d c)^2B = Q_{-cK} \d c B.\end{equation}
Now consider the MNT ghost brane. There are three possible elements of the 
cohomology at level $2$:
\begin{equation}K (\d c)^2 B,\ \ \ \ B(\d c)^2 K,\ \ \ \ \d c KB\d c.
\end{equation}
The first two states vanish assuming $(\d c)^2 = 0$, but are also 
trivial in the cohomology:
\begin{eqnarray}K (\d c)^2 B\lineup = Q_{-cK-Kc+cKBc} \,K\d c B,\nonumber\\
B (\d c)^2 K \lineup = Q_{-cK-Kc+cKBc}\,B \d c K.
\end{eqnarray}
The third state $\d c KB \d c$ does not vanish, and is not 
$Q_{\Psi_0}$-closed if we assume only the defining relations \eq{KB_id} 
and \eq{c_id} of the $KBc$ subalgebra. However, if we assume auxiliary 
identities $\d c KB\d c$ is a nontrivial element of the cohomology
\begin{equation}Q_{-cK-Kc+cKBc}\,\d c KB \d c = 0,\ \ \ \  
\d c KB\d c\neq Q_{-ck-Kc+cKBc}(\mathrm{something})\end{equation}
Therefore $\d c KB \d c$ generates a physically nontrivial
deformation of the MNT ghost brane background. However, the resulting 
solutions require auxiliary identities to satisfy the equations of motion,
and our more limited goal is to classify solutions which satisfy the equations
of motion only by virtue of the defining relations of the $KBc$ subalgebra 
\eq{KB_id} and \eq{c_id}. Then for our purposes $\d c KB \d c$ is not 
$Q_{\Psi_0}$-closed, and in all cases \eq{Ueq2} has a solution $U_{n+1}$. 
This completes the proof.

\begin{table}[t]
\renewcommand{\arraystretch}{1.0}
\begin{center}
\begin{tabular}{|c|c|c|c|c|}
\hline 
& \ \ \ level $-1$\ \ \  & \ \ \ \ level $0$\ \ \ \  & \ \ \ \ level $1$\ \ \ \  & \ \ \ \ 
level $2$ \ \ \ \  \\
\hline 
 & & & & \\
$\displaystyle \begin{matrix}\mathrm{Residual\ perturbative}\\ \mathrm{vacuum}\end{matrix}$  
& --- & $\displaystyle \begin{matrix}cKBc \\ 
Bc \d c \\ c\d c B\end{matrix}$ & --- & --- \\
& & & & \\
\hline
& & & & \\
$\displaystyle \begin{matrix}\mathrm{Residual\ tachyon}\\ \mathrm{vacuum}\end{matrix}$  
& $c$ & $\displaystyle \begin{matrix} 
Bc \d c \\ c\d c B\end{matrix}$ & --- & --- \\
& & & &\\
\hline
& & & & \\
$\displaystyle \begin{matrix}\mathrm{Residual\ conjugate}\\ \mathrm{tachyon\ vacuum}\end{matrix}$  
& $c$ & $\displaystyle \begin{matrix} 
Bc \d c \\ c\d c B\end{matrix}$ & --- & --- \\
& & & & \\
\hline 
& & & & \\
MNT\ ghost\ brane & --- & $\displaystyle \begin{matrix} 
Bc \d c \\ c\d c B\end{matrix}$ & --- & $\d c KB \d c$ \\
& & & & \\
\hline
\end{tabular}
\end{center}
\caption{\label{tab:coh} Elements of the ghost number $1$ cohomology
in the $KBc$ algebra (extended with auxiliary identities) around all four 
residual solutions. We take the 
kinetic operator $Q_{\Psi_0}$ around the level $0$ representatives of 
these gauge orbits. The cohomology element $c$ represents a deformation
of the residual tachyon vacuum into the tachyon vacuum. Likewise, $cKBc$ 
represents a deformation of the residual perturbative 
vacuum into the perturbative vacuum. The remaining cohomology elements follow
from auxiliary identities. The fields $Bc\d c$ 
and $c\d c B$ generate marginal deformations \eq{aux} which connect 
the residual solutions inside a single moduli space. The field 
$\d c KB \d c$ represents an additional 
deformation of the MNT ghost brane whose interpretation is unclear.} 
\end{table}

An byproduct of our proof is a classification of the physical
cohomology in the $KBc$ subalgebra, supplemented with auxiliary identities, 
around all four residual solutions. We list these in 
table \ref{tab:coh}. Note that the perturbative vacuum and (of course)
the tachyon vacuum have no ghost number $1$ cohomology in the $KBc$ 
subalgebra since there are no on-shell vertex operators in the universal 
sector.

The reader may ask what can be said about the classification of solutions
in the $KBc$ subalgebra additionally assuming the full set of relations 
satisfied by $c$. In this case we would have to examine the cohomology 
at positive levels around all of the solutions in \eq{aux}. 
There is a possibility of enhanced 
cohomology whenever $\gamma_1+\gamma_2$ is an integer. Moreover, for these
solutions theorem \ref{thm:nogo} would have to be carefully reconsidered, 
as its proof is based on the functional relation \eq{KEOM} which assumes 
only the basic relations \eq{c_id}. Therefore, we have not excluded the 
possibility that auxiliary identities could produce physically 
interesting solutions, for example multibranes. 

\end{appendix}

\end{document}